\documentclass[12pt]{article}
\usepackage[left=35mm,right=35mm,top=35mm,bottom=35mm]{geometry}
\usepackage{amsmath}
\usepackage{amsthm}
\usepackage{amssymb}
\usepackage{amsfonts}
\usepackage{mathrsfs}
\usepackage{color}

\DeclareMathOperator*{\slim}{s-lim}

\DeclareMathOperator*{\arcsinh}{arcsinh}

\makeatletter

\@addtoreset{equation}{section}
\makeatother
\newtheorem{The}{Theorem}[section]
\newtheorem{Ass}[The]{Assumption}

\newtheorem{Lem}[The]{Lemma}

\title{Inverse scattering for repulsive potential and strong singular interactions}
\author{Atsuhide ISHIDA\\
\\
\normalsize E-mail: aishida@rs.tus.ac.jp\\
Department of Liberal Arts, Faculty of Engineering,
 \\Tokyo University of Science\\
\normalsize 6-3-1 Niijuku, Katsushika-ku,Tokyo 125-8585, Japan\\
Alfr\'ed R\'enyi Institute of Mathematics, Re\'altanoda utca 13-15, 1053 Budapest, Hungary
}
\date{}

\begin{document}
\begin{flushleft}
{\Large \bf Inverse scattering for repulsive potential and strong singular interactions}
\end{flushleft}

\begin{flushleft}
{\large Atsuhide ISHIDA}\\
{Katsushika Division, Institute of Arts and Sciences, Tokyo University of Science, 6-3-1 Niijuku, Katsushika-ku, Tokyo 125-8585, Japan\\ 
Alfr\'ed R\'enyi Institute of Mathematics, Re\'altanoda utca 13-15, Budapest 1053, Hungary\\
Email: aishida@rs.tus.ac.jp\\
}
\end{flushleft}

%abstract---------------------------------------------------
\begin{abstract}
\noindent
In a previous work of 2014 on a quantum system governed by the repulsive Hamiltonian, the author proved uniqueness for short-range interactions described by a scattering operator consisting of regular and singular parts. In this paper, the singular part is assumed to have much stronger singularities and the same uniqueness theorem is proved. By applying the time-dependent method invented by Enss and Weder in 1995, the high-velocity limit for a wider class of the scattering operator with stronger singularities also uniquely determines the interactions of a multi-dimensional system. 
\end{abstract}


%%%%%%%%%%%%%%%%%%%%%%%%
\section{Introduction\label{introduction}}

We investigate multidimensional inverse scattering for the Schr\"odinger operator that has a repulsive term. Throughout this paper, we assume space has dimensions $n\geqslant2$. The free dynamics is described by the free Hamiltonian
\begin{equation}
H_0=p^2-x^2,\label{free}
\end{equation}
called the repulsive Hamiltonian, acting on $L^2(\mathbb{R}^n)$, where $p=-{\rm i}(\partial_{x_1},\ldots,\partial_{x_n})=-{\rm i}\nabla$ is the momentum operator, $x=(x_1,\ldots,x_n)\in\mathbb{R}^n$ is the position of the particle, and $x^2$ means $|x|^2$. Thus, $p^2=|p|^2=-\sum_{j=1}^n\partial_{x_j}^2=-\Delta$ is the negative of the Laplacian. If the repulsive term $-x^2$ is replaced by $+x^2$, $H_0$ becomes the well-known harmonic oscillator Hamiltonian and, in this case, the particle system retains some of the bound states but does not have scattering states.\par

The interactional potential $V=V(x)$ is a multiplicative operator and decomposes as $V=V^{\rm sing}+V^{\rm reg}\in\mathscr{V}^{\rm sing}+\mathscr{V}^{\rm reg}$, which satisfies the following Assumption \ref{ass}.

\begin{Ass}\label{ass}
The real-valued function $V^{\rm sing}\in\mathscr{V}^{\rm sing}$ is compactly supported in $\mathbb{R}^n$ and $V^{\rm sing}\in L^q(\mathbb{R}^n)$ where
\begin{equation}
q
\begin{cases}
\ =2 &\mbox{\rm if}\quad n\leqslant3,\\
\ >2 &\mbox{\rm if}\quad n=4,\\
\ =n/2 &\mbox{\rm if}\quad n\geqslant5.
\end{cases}\label{q}
\end{equation}
The real-valued function $V^{\rm reg}\in\mathscr{V}^{\rm reg}$ belongs to $C^1(\mathbb{R}^n)$ and satisfies
\begin{equation}
|\partial_x^\beta V^{\rm reg}(x)|\lesssim_\beta\langle x\rangle^{-\epsilon-|\beta|}\label{reg}
\end{equation}
for $\beta\in(\mathbb{N}\cup\{0\})^n$ and $|\beta|\leqslant1$, where $\epsilon>0$, $\langle x\rangle=\sqrt{1+x^2}$ and $A\lesssim_\beta B$ means that there exists a constant $C_\beta>0$ such that $A\leqslant C_\beta B$. If the constant has no specific dependence, we write $A\lesssim B$.
\end{Ass}

If $q=2$ for $n\leqslant3$ and $q>n/2$ for $n\geqslant4$, then $V^{\rm sing}$ is well-known to be relatively $p^2$-compact (see \cite[Theorem 8.19]{Sc} for example); moreover the wave operators
\begin{equation}
W^\pm=\slim_{t\rightarrow\pm\infty}e^{{\rm i}tH}e^{-{\rm i}tH_0}\label{wave_op}
\end{equation}
exist and are asymptotically complete by the result given in \cite{BoCaHaMi} (see also \cite{Ko}), where $H=H_0+V$. If $q=n/2$ for $n\geqslant5$, it is known that $V^{\rm sing}$ is relatively $p^2$-bounded infinitesimally (see \cite[Theorem X. 21]{ReSi2}). Although its relatively $p^2$-compactness also holds by applying Hardy-Littlewood-Sobolev inequality (\cite[VIII 4.2]{St2}) even in this instance, we instead give the direct proof of the existence of \eqref{wave_op} in section \ref{high_dim_case}. We thus define the scattering operator such that
\begin{equation}
S(V)=(W^+)^*W^-.
\end{equation}
The main theorem in this paper is the following statement.

\begin{The}\label{the1}
Let $V_1,V_2\in \mathscr{V}^{\rm sing}+\mathscr{V}^{\rm reg}$. If $S(V_1)=S(V_2)$, then $V_1=V_2$.
\end{The}

The previous work \cite{Is1} proved this theorem assuming $V^{\rm sing}$ is compactly supported and $V^{\rm sing}\in L^{\tilde{q}}(\mathbb{R}^n)$ where
\begin{equation}
\tilde{q}
\begin{cases}
\ >2 &\mbox{\rm if}\quad n\leqslant4,\\
\ >n/2 &\mbox{\rm if}\quad n\geqslant5.
\end{cases}
\end{equation}
We note that if $V^{\rm sing}\in L^{\tilde{q}}(\mathbb{R}^n)$ then $V^{\rm sing}\in L^q(\mathbb{R}^n)$ by the H\"older inequality because $V^{\rm sing}$ is compactly supported. Therefore, Theorem \ref{the1} extends the result given in \cite{Is1} in which $q=2$ for $n\leqslant3$ and $q=n/2$ for $n\geqslant5$. If $n=2$ or $3$, the function
\begin{equation}
c|x|^{-n/2+\epsilon}F(|x|\leqslant1)\label{coulomb}
\end{equation}
belongs to $\mathscr{V}^{\rm sing}$ for any $\epsilon>0$ and $c\in\mathbb{R}$ because
\begin{equation}
\int_{|x|\leqslant1}|x|^{-n+2\epsilon}{\rm d}x=\omega_n\int_0^1r^{2\epsilon-1}{\rm d}r<\infty
\end{equation}
where $\omega_n$ is the surface area of the unit ball; that is, local Coulomb-like singularities are permissible for $n=3$. However, to guarantee that the function \eqref{coulomb} belongs to $L^{\tilde{q}}(\mathbb{R}^n)$, we must assume $\epsilon>n(1/2-1/\tilde{q})$ because
\begin{equation}
\int_{|x|\leqslant1}|x|^{-\tilde{q}n/2+\tilde{q}\epsilon}{\rm d}x=\omega_n\int_0^1r^{-\tilde{q}n/2+\tilde{q}\epsilon+n-1}{\rm d}r.
\end{equation}
This says that Coulomb-like singularities are permitted for $n=3$ only if $\tilde{q}<3$.\par
Since the Enss--Weder time-dependent method was devised, many authors have applied it to establish the uniqueness of the interaction potentials for various quantum models. References \cite{AdKaKaTo}, \cite{AdFuIs}, \cite{AdMa}, \cite{Is3}, \cite{Ni1}, \cite{Ni2}, \cite{VaWe}, and \cite{We1} investigated models with external electric fields, whereas references \cite{Is4} and \cite{Ju} studied fractional and relativistic Laplacians, and \cite{Is5} studied time-dependent harmonic oscillators. References \cite{Wa1}, \cite{Wa2}, \cite{Wa3}, \cite{We2} and \cite{We3} applied the method to non-linear Schr\"odinger and Hartree--Fock equations. As for repulsive Hamiltonians, Theorem \ref{the1} was first proved in \cite{Ni3} assuming $V\in C^\infty(\mathbb{R}^n)$ satisfies
\begin{equation}
|\partial_x^\beta V(x)|\lesssim_\beta\langle x\rangle^{-1/2-\epsilon-|\beta|}
\end{equation} 
for $\epsilon>0$ and any $\beta\in(\mathbb{N}\cup\{0\})^n$ without any singularities. Later \cite{Is1} treated the singular part $V^{\rm sing}\in L^{\tilde{q}}(\mathbb{R}^n)$ with compact support and relaxed the regularities to $C^1(\mathbb{R}^n)$ and the decay assumption to \eqref{reg}.\par
The time evolution of Schr\"odinger operator with a repulsive term has interesting features from both physical and mathematical aspects. By solving the Newton equation ${\rm d}^2{x(t)}/{\rm d}t^2=4x(t)$, we find that the classical orbit of the particle has exponential order $x(t)=O(e^{2t})$ as $t\rightarrow\infty$. On the basis of this observation, \cite{BoCaHaMi} proved the asymptotically completeness of the wave operators subject to
\begin{equation}
|V(x)|\lesssim(\log\langle x\rangle)^{-1-\epsilon}\label{log}
\end{equation}
for $\epsilon>0$ and a $p^2$-compact singular term. The motion of the classical orbit implies that the decay in \eqref{log} is short-range. Indeed, \cite{Is2} constructed concrete examples of potential functions that have slower decay than $(\log\langle x\rangle)^{-1}$ and of which the wave operators do not exist. Recently, with the studies of \cite{It1} and \cite{It2} on stationary scattering and the limiting absorption principle, spectral and scattering theory for repulsive Hamiltonians is still making progress.\par
The free time evolution $e^{-{\rm i}tH_0}$, called the Mehler formula and introduced by H\"ormander \cite{Ho}, has a useful representation. Let us here state this formula for repulsive Hamiltonians. As is well-known, the version for the harmonic oscillator involves trigonometric functions. In contrast the repulsive version involves hyperbolic functions. For $t\not=0$, the unitary propagator $e^{-{\rm i}tH_0}$ is represented as
\begin{equation}
e^{-{\rm i}tH_0}=\mathscr{M}(\tanh(2t)/2)\mathscr{D}(\sinh(2t)/2)\mathscr{F}\mathscr{M}(\tanh(2t)/2)\label{mehler1}%Mehler formula
\end{equation}
where $\mathscr{M}$ and $\mathscr{D}$ denote multiplication and dilation, with
\begin{gather}
\mathscr{M}(t)\phi(x)=e^{{\rm i}x^2/(4t)}\phi(x),\\
\mathscr{D}(t)\phi(x)=(2{\rm i}t)^{-n/2}\phi(x/(2t)),
\end{gather}
and $\mathscr{F}$ denotes the Fourier transform over $L^2(\mathbb{R}^n)$. A straightforward calculation yields
\begin{gather}
\mathscr{D}(\sinh(2t)/2)={\rm i}^{n/2}\mathscr{D}(\cosh(2t)/2)\mathscr{D}(\tanh(2t)/2),\\
\mathscr{M}(\tanh(2t)/2)\mathscr{D}(\cosh(2t)/2)\mathscr{M}(-\tanh(2t)/2)\nonumber\\
=\mathscr{M}(\coth(2t)/2)\mathscr{D}(\cosh(2t)/2)
\end{gather}
and
\begin{equation}
e^{-{\rm i}tH_0}={\rm i}^{n/2}\mathscr{M}(\coth(2t)/2)\mathscr{D}(\cosh(2t)/2)e^{-{\rm i}\tanh(2t)p^2/2},\label{mehler2}
\end{equation}
because
\begin{equation}
e^{-{\rm i}tp^2}=\mathscr{M}(t)\mathscr{D}(t)\mathscr{F}\mathscr{M}(t).
\end{equation}

Using the Plancherel formula for the Radon transform \cite[Theorem 2.17 in Chap. 1]{He} as in the proof of \cite[Theorem 1.1]{EnWe}, the following reconstruction theorem yields the proof of Theorem \ref{the1}. For the remainder of this paper, we therefore focus on proving Theorem \ref{the2}. Regarding notation throughout this paper, $(\cdot,\cdot)$ denotes the $L^2$-scalar product, $\|\cdot\|$ denotes the $L^2$-norm, and $\|\cdot\|_\mathscr{B}$ denotes the $L^2$-operator norm. For the $L^r$-norm and $L^r$-operator norm with $r\not=2$, we write $\|\cdot\|_r$ and $\|\cdot\|_{\mathscr{B}(L^r)}$.

\begin{The}\label{the2}
Let $\Phi_0\in\mathscr{S}(\mathbb{R}^n)$ such that $\mathscr{F}\Phi_0\in C_0^\infty(\mathbb{R}^n)$ and $\hat{v}=v/|v|$ be the normalization of $v\in\mathbb{R}^n$. Let $\Phi_v=e^{{\rm i}v\cdot x}\Phi_0$ and $\Psi_v$ have the same properties. Then
\begin{gather}
\lim_{|v|\rightarrow\infty}|v|({\rm i}[S,p_j]\Phi_v,\Psi_v)=\int_{-\infty}^\infty\left\{(V^{\rm sing}(x+\hat{v}t)p_j\Phi_0,\Psi_0)-(V^{\rm sing}(x+\hat{v}t)\Phi_0,p_j\Psi_0)\right.\nonumber\\
+\left.({\rm i}(\partial_{x_j}V^{\rm reg})(x+\hat{v}t)\Phi_0,\Psi_0)\right\}{\rm d}t/2
\end{gather}
holds where $p_j$ is the $j$th component of $p$.
\end{The}

%%%%%%%%%%%%%%%%%%%%%%%%
\section{High-dimensional case}\label{high_dim_case}

In this section, we consider instances for which $n\geqslant5$ and give a proof of Theorem \ref{the2}. For this purpose, we give the direct proof of the existence of the wave operators as we mentioned before. We then can demonstrate the proof of Theorem \ref{the2} in the same manner as given in \cite{Is1}. We now start from the self-adjointness of $H$.

\begin{Lem}\label{lem3}
Let $V^{\rm sing}\in \mathscr{V}^{\rm sing}$. then $V^{\rm sing}$ is $H_0$-bounded infinitesimally.
\end{Lem}

\begin{proof}
Let $\chi\in C_0^{\infty}(\mathbb{R}^n)$ be such that $\chi=1$ near the support of $V^{\rm sing}$. From the proof given in \cite[Lemma 2.3]{BoCaHaMi}, we have
\begin{gather}
\chi({\rm i}+H_0)^{-1}=({\rm i}+p^2)^{-1}B_1,\quad\chi({\rm i}+p^2)^{-1}=({\rm i}+H_0)^{-1}B_2,
\end{gather}
where $B_1$ and $B_2$ are bounded operators. Let $\phi\in\mathscr{D}(H_0)$. For any $\epsilon>0$, there exists $C_\epsilon>0$ such that
\begin{gather}
\|V^{\rm sing}\phi\|=\|V^{\rm sing}\chi^2({\rm i}+H_0)^{-1}({\rm i}+H_0)\phi\|\nonumber\\
\leqslant\epsilon\|p^2\chi({\rm i}+p^2)^{-1}B_1({\rm i}+H_0)\phi\|+C_\epsilon\|\chi({\rm i}+p^2)^{-1}B_1({\rm i}+H_0)\phi\|\nonumber\\
=\epsilon\|p^2\chi({\rm i}+p^2)^{-1}B_1({\rm i}+H_0)\phi\|+C_\epsilon\|({\rm i}+H_0)^{-1}B_2B_1({\rm i}+H_0)\phi\|
\end{gather}
because $V^{\rm sing}$ is $p^2$-bounded infinitesimally. We find that $p^2\chi({\rm i}+p^2)^{-1}B_1$ and $({\rm i}+H_0)^{-1}B_2B_1({\rm i}+H_0)$ are bounded. This completes the proof.
\end{proof}

\begin{Lem}\label{lem4}
The wave operators \eqref{wave_op} exist.
\end{Lem}

\begin{proof}
We prove $W^+$ only because the proof of the existence of $W^-$ is similar. According to the Cook-Kuroda method \cite[Theorem XI. 4]{ReSi3}, it suffices to prove
\begin{equation}
\int_1^\infty\|V(x)e^{-{\rm i}tH_0}\phi\|{\rm d}t<\infty
\end{equation}
for $\phi\in C_0^\infty(\mathbb{R}^n)$. By assumption \eqref{reg} and \cite[Proposition 2.2]{Is2}, we easily find that
\begin{equation}
\int_1^\infty\|V^{\rm reg}(x)e^{-{\rm i}tH_0}\phi\|{\rm d}t<\infty.
\end{equation}
By \cite[Lemma 6]{Ni3}, the relation
\begin{equation}
e^{{\rm i}tH_0}xe^{-{\rm i}tH_0}=\cosh(2t)x+\sinh(2t)p
\end{equation}
holds. We therefore have
\begin{gather}
\|V^{\rm sing}(x)e^{-{\rm i}tH_0}\phi\|=\|V^{\rm sing}(\cosh(2t)x+\sinh(2t)p)\phi\|\nonumber\\
=\|V^{\rm sing}(\sinh(2t)p)e^{{\rm i}\coth(2t)x^2/2}\phi\|.\label{lem4_1}
\end{gather}
The H\"older inequality \cite[Theorem III. 1]{ReSi1} and the Hausdorff--Young inequality \cite[Theorem IX. 8]{ReSi1} imply that the right-hand side of \eqref{lem4_1} is less than or equal to
\begin{gather}
\|V^{\rm sing}(\sinh(2t)\xi)\|_{n/2}\|\mathscr{F}e^{{\rm i}\coth(2t)x^2/2}\phi\|_{2n/(n-4)}\nonumber\\
\lesssim|\sinh(2t)|^{-2}\|V^{\rm sing}\|_{n/2}\|\phi\|_{2n/(n+4)}.
\end{gather}
This implies that
\begin{equation}
\int_1^\infty\|V^{\rm sing}(x)e^{-{\rm i}tH_0}\phi\|{\rm d}t<\infty
\end{equation}
and that $W^+$ exists.
\end{proof}

%%%%%%%%%%%%%%%%%%%%%%%%
\section{Low-dimensional case}\label{low_dim_case}
In this section, we consider instances for which $n\leqslant3$ and $V^{\rm sing}\in L^2(\mathbb{R}^n)$ and give a proof of Theorem \ref{the2}. As mentioned before, the wave operators exist and are asymptotically complete by the result given in \cite{BoCaHaMi}. To prove Theorem \ref{the2}, the key propagation estimate is given by Lemma \ref{lem5}:

\begin{Lem}\label{lem5}
Let $\Phi_v$ be as in Theorem \ref{the2}. Then
\begin{equation}
\int_{-\infty}^\infty\|V^{\rm sing}(x)e^{-{\rm i}tH_0}\Phi_v\|{\rm d}t=O(|v|^{-1})
\end{equation}
hold as $|v|\rightarrow\infty$ for $V^{\rm sing}\in\mathscr{V}^{\rm sing}$.
\end{Lem}

Before proving Lemma \ref{lem5}, we prepare Lemma \ref{lem6}. In a previous work \cite{Is1}, the H\"ormander--Mikhlin inequality for the Fourier multipliers (see \cite[Chap. IV, Theorem 3]{St1} for example) plays an important role in proving the same statement for Lemma \ref{lem5} and therefore $\tilde{q}>2$ must be assumed. In our case, we cannot rely on the H\"ormander--Mikhlin inequality because $q=2$. Instead, our strategy exploits a Carlson--Beurling-type inequality with a uniform scale invariance. 

\begin{Lem}\label{lem6}
Let $m\in C^\infty(\mathbb{R}^n)\cap H^2(\mathbb{R}^n)$. Then
\begin{equation}
\sup_{\lambda\in\mathbb{R}}\|m(\lambda p)\|_{\mathscr{B}(L^\infty)}\lesssim\|m\|^{1-n/4}\|m\|_{\dot{H}^2}^{n/4}\label{lem6_0}
\end{equation}
holds where $\dot{H}^2=\dot{H}^2(\mathbb{R}^n)$ denotes the homogeneous Sobolev space of order $2$.
\end{Lem}

\begin{proof}
By virtue of the Carlson--Beurling inequality \cite[Chap. 1, Theorem 3.1]{PhViLa}, we have
\begin{equation}
\|m(\lambda p)\|_{\mathscr{B}(L^\infty)}\lesssim\|m(\lambda x)\|^{1-n/4}\|m(\lambda x)\|_{\dot{H}^2}^{n/4},\label{lem6_1}
\end{equation}
noting that the constant in front of the right-hand side of \eqref{lem6} is independent of $\lambda\in\mathbb{R}$. We assume here that $m\in C_0^\infty(\mathbb{R}^n)$. We clearly have
\begin{equation}
\|m(\lambda x)\|^2=|\lambda|^{-n}\|m\|^2.\label{lem6_2}
\end{equation}
We also have
\begin{equation}
\int_{\mathbb{R}^n}e^{-{\rm i}x\cdot\xi}m(\lambda x){\rm d}x=|\lambda|^{-n
}\int_{\mathbb{R}^n}e^{-{\rm i}y\cdot\xi/\lambda}m(y){\rm d}y
\end{equation}
and
\begin{gather}
\|m(\lambda x)\|_{\dot{H}^2}^2=\||\xi|^2\mathscr{F}m(\lambda x)\|^2=|\lambda|^{-2n}\int_{\mathbb{R}^n}|\xi|^4\left|\int_{\mathbb{R}^n}e^{-{\rm i}y\cdot\xi/\lambda}m(y){\rm d}y\right|^2{\rm d}\xi/(2\pi)^n\nonumber\\
=|\lambda|^{-n+4}\int_{\mathbb{R}^n}|\eta|^4\left|\int_{\mathbb{R}^n}e^{-{\rm i}y\cdot\eta}m(y){\rm d}y\right|^2{\rm d}\eta/(2\pi)^n
=|\lambda|^{-n+4}\|m\|_{\dot{H}^2}^2\label{lem6_3}
\end{gather}
by changing the variable $\eta=\xi/\lambda$. Noting that $(-n/2)(1-n/4)+(-n/2+2)(n/4)=0$, we have \eqref{lem6_0} from \eqref{lem6_1}, \eqref{lem6_2} and \eqref{lem6_3}. If $m\in C^\infty(\mathbb{R}^n)\cap H^2(\mathbb{R}^n)$ generally, we take $m_k\in C_0^\infty(\mathbb{R}^n)$ such that $m_k\rightarrow m$ as $k\rightarrow m$ in $H^2(\mathbb{R}^n)$. Applying the limit in \eqref{lem6_0} for $m_k$, we have \eqref{lem6_0} for $m$ as in the proof of \cite[Chap. 1, Theorem 3.1]{PhViLa}.
\end{proof}

\begin{proof}[Proof of Lemma \ref{lem5}]
We separate the integral such that
\begin{equation}
\int_{-\infty}^\infty=\int_{|t|<1}+\int_{|t|\geqslant1}
\end{equation}
and we already have
\begin{equation}
\int_{|t|<1}\|V^{\rm sing}(x)e^{-{\rm i}tH_0}\Phi_v\|{\rm d}t=O(|v|^{-1})
\end{equation}
in the proof of \cite[Proposition 2.2]{Is1}. Let us consider the integral over $|t|\geqslant1$. By the Mehler formula, we have
\begin{gather}
e^{-{\rm i}tH_0}\Phi_v=\mathscr{M}(\tanh(2t)/2)e^{-{\rm i}\sinh(2t)v\cdot p}\mathscr{D}(\sinh(2t)/2)\mathscr{F}\mathscr{M}(\tanh(2t)/2)\Phi_0\nonumber\\
=e^{-{\rm i}\sinh(2t)v\cdot p}e^{{\rm i}\cosh(2t)\sinh(2t)v^2/2}e^{{\rm i}\cosh(2t)x\cdot v}e^{-{\rm i}tH_0}\Phi_0
\end{gather}
and
\begin{gather}
|V^{\rm sing}(x)e^{-{\rm i}tH_0}\Phi_v\|=|V^{\rm sing}(x+\sinh(2t)v)e^{-{\rm i}tH_0}\Phi_0\|\nonumber\\
=\|V^{\rm sing}(\sinh(2t)(x+v))\mathscr{F}\mathscr{M}(\tanh(2t)/2)\Phi_0\|\leqslant I_1+I_2
\end{gather}
where
\begin{align}
I_1&=\|V^{\rm sing}(\sinh(2t)(x+v))\langle p/\sinh(2t)\rangle^{-2}F(|x|<|v|/2)\|_\mathscr{B}\nonumber\\
&\quad\times\|\langle x/\sinh(2t)\rangle^2\Phi_0\|,\\
I_2&=\|V^{\rm sing}(\sinh(2t)(x+v))\|\nonumber\\
&\quad\times\|\langle p/\sinh(2t)\rangle^{-2}F(|x|\geqslant|v|/2)\mathscr{F}\mathscr{M}(\tanh(2t)/2)\langle x/\sinh(2t)\rangle^2\Phi_0\|_\infty
\end{align}
as in the proof of \cite[Proposition 2.2]{Is1}. When $|x|<|v|/2$, we have $|\sinh(2t)(x+v)|>|\sinh(2t)v|/2$ and
\begin{gather}
\|V^{\rm sing}(\sinh(2t)(x+v))\langle p/\sinh(2t)\rangle^{-2}F(|x|<|v|/2)\|_\mathscr{B}\nonumber\\
\leqslant\|V^{\rm sing}(\sinh(2t)(x+v))\langle p/\sinh(2t)\rangle^{-2}F(|\sinh(2t)(x+v)|>|\sinh(2t)v|/2)\|_\mathscr{B}\nonumber\\
=\|V^{\rm sing}(x)\langle p\rangle^{-2}F(|x|>|\sinh(2t)v|/2)\|_\mathscr{B}.
\end{gather}
We therefore have
\begin{gather}
\int_{|t|\geqslant1}I_1{\rm d}t\lesssim\|\langle x\rangle^2\Phi_0\|\int_{|t|\geqslant1}\|V^{\rm sing}(x)\langle p\rangle^{-2}F(|x|>|\sinh(2t)v|/2)\|_\mathscr{B}{\rm d}t\nonumber\\
\lesssim|v|^{-1}\int_{|v|\sinh2}^\infty\|V^{\rm sing}(x)\langle p\rangle^{-2}F(|x|>\tau/2)\|_\mathscr{B}{\rm d}\tau.
\end{gather}
We here take $\chi\in C^\infty(\mathbb{R}^n)$ such that $\chi(x)=1$ if $|x|\geqslant1$ and $\chi(x)=0$ if $|x|\leqslant1/2$, and obtain
\begin{gather}
\|V^{\rm sing}(x)\langle p\rangle^{-2}F(|x|>\tau/2)\|_\mathscr{B}\leqslant\|V^{\rm sing}(x)\langle p\rangle^{-2}\chi(2x/\tau)\|_\mathscr{B}\nonumber\\
\lesssim\|V^{\rm sing}(x)\chi(2x/\tau)\langle p\rangle^{-2}\|_\mathscr{B}+\tau^{-1}\|V^{\rm sing}(x)(\nabla\chi)(2x/\tau)\langle p\rangle^{-2}\|_\mathscr{B}\nonumber\\
+\tau^{-2}\|V^{\rm sing}(x)\langle p\rangle^{-2}(\Delta\chi)(2x/\tau)\|_\mathscr{B}\label{lem5_1}
\end{gather}
by calculating the commutator $[\langle p\rangle^{-2},\chi(2x/\tau)]$. Noting that $V^{\rm sing}$ is compactly supported and that the first and second terms of the right-hand side of \eqref{lem5_1} are
\begin{equation}
\int_{|v|\sinh2}^\infty(\|V^{\rm sing}(x)\chi(2x/\tau)\langle p\rangle^{-2}\|_\mathscr{B}+\tau^{-1}\|V^{\rm sing}(x)(\nabla\chi)(2x/\tau)\langle p\rangle^{-2}\|_\mathscr{B}){\rm d}\tau=0
\end{equation}
for $|v|\gg1$, we have
\begin{equation}
\int_{|t|\geqslant1}I_1{\rm d}t\lesssim|v|^{-1}\int_{|v|\sinh2}^\infty\tau^{-2}{\rm d}\tau=O(|v|^{-2}).
\end{equation}
We next consider the integral of $I_2$. Applying Lemma \ref{lem6} for $m(x)=\langle x\rangle^{-2}$ and $\lambda=1/\sinh(2t)$, we have
\begin{equation}
\|\langle p/\sinh(2t)\rangle^{-2}\|_{\mathscr{B}(L^\infty)}\lesssim1.\label{lem5_2}
\end{equation}
We write
\begin{gather}
\mathscr{F}\mathscr{M}(\tanh(2t)/2)\langle x/\sinh(2t)\rangle^{2}\Phi_0\nonumber\\
=\int_{\mathbb{R}^n}e^{-{\rm i}x\cdot y}e^{{\rm i}y^2/(2\tanh(2t))}\langle y/\sinh(2t)\rangle^{2}\Phi_0(y){\rm d}y/(2\pi)^{n/2}.
\end{gather}
Using the relation $e^{-{\rm i}x\cdot y}=\langle x\rangle^{-2}(1+{\rm i}x\cdot \nabla_y)e^{-{\rm i}x\cdot y}$ and integrating by parts, we have
\begin{gather}
\mathscr{F}\mathscr{M}(\tanh(2t)/2)\langle x/\sinh(2t)\rangle^{2}\Phi_0=\langle x\rangle^{-2}\mathscr{F}\mathscr{M}(\tanh(2t)/2)\langle x/\sinh(2t)\rangle^{2}\Phi_0\nonumber\\
+(1/\tanh(2t))\langle x\rangle^{-2}x\cdot\mathscr{F}x\mathscr{M}(\tanh(2t)/2)\langle x/\sinh(2t)\rangle^{2}\Phi_0\nonumber\\
-{\rm i}\langle x\rangle^{-2}x\cdot\mathscr{F}\mathscr{M}(\tanh(2t)/2)\nabla_x\langle x/\sinh(2t)\rangle^{2}\Phi_0.\label{lem5_3}
\end{gather}
The Hausdorff--Young inequality implies that the Fourier transform is bounded from $L^1(\mathbb{R}^n)$ to $L^\infty(\mathbb{R}^n)$, and we therefore have
\begin{gather}
\|F(|x|\geqslant|v|/2)\mathscr{F}\mathscr{M}(\tanh(2t)/2)\langle x/\sinh(2t)\rangle^{2}\Phi_0\|_\infty\nonumber\\
\lesssim |v|^{-2}\|\langle x\rangle^2\Phi_0\|_1+|v|^{-1}(\|\langle x\rangle^3\Phi_0\|_1+\|\langle x\rangle^2\nabla\Phi_0\|_1).\label{lem5_4}
\end{gather}
By \eqref{lem5_2}, \eqref{lem5_4} and
\begin{equation}
\|V^{\rm sing}(\sinh(2t)(x+v))\|=|\sinh(2t)|^{-n/2}\|V^{\rm sing}\|
\end{equation}
we consequently have
\begin{equation}
\int_{|t|\geqslant1}I_2{\rm d}t\lesssim|v|^{-1}\int_{|t|\geqslant1}|\sinh(2t)|^{-n/2}{\rm d}t=O(|v|^{-1}).
\end{equation}
This completes the proof.
\end{proof}

By virtue of \cite[Proposition 2.4]{Is1}, the following propagation estimate holds.
\begin{Lem}\label{lem7}
Let $\Phi_v$ be as in Theorem \ref{the2}. Then
\begin{equation}
\int_{-\infty}^\infty\|\{V^{\rm reg}(x)-V^{\rm reg}(\sinh(2t)v)\}e^{-{\rm i}tH_0}\Phi_v\|{\rm d}t=O(|v|^{-(1+\epsilon)/2})
\end{equation}
holds as $|v|\rightarrow\infty$ for $V^{\rm reg}\in\mathscr{V}^{\rm reg}$.
\end{Lem}

As in \cite{Is1}, we introduce the modified wave operators
\begin{equation}
\Omega_v^\pm=\slim_{t\rightarrow\pm\infty}e^{{\rm i}tH}U_v(t),\quad U_v(t)=e^{-{\rm i}tH_0}e^{-{\rm i}\int_0^tV^{\rm reg}(\sinh(2\tau)v){\rm d}\tau}.
\end{equation}
The Dollard-type modified wave operators were applied even for short-range inverse scattering under the electric fields by \cite{Ni1} and \cite{Ni2}. Later, the Graf-type \cite{Gr} (or Zorbas-type \cite{Zo}) modified wave operators were first introduced in \cite{AdMa} for inverse scattering in the Stark effect instead of the Dollard-type. Since then, the Graf-type modification for inverse scattering under the electric fields has been the standard method (see also \cite{AdKaKaTo}, \cite{AdFuIs} and \cite{Is3}). Our modifier $e^{-{\rm i}\int_0^tV^{\rm reg}(\sinh(2\tau)v){\rm d}\tau}$ also commutes with any operator as with the Graf-type.\par
The following Lemma \ref{lem8} can be proven as in \cite[Proposition 2.5]{Is1} by virtue of Lemmas \ref{lem5} and \ref{lem7}.

\begin{Lem}\label{lem8}
Let $\Phi_v$ be as in Theorem \ref{the2}. Then
\begin{equation}
\sup_{t\in\mathbb{R}}\|(e^{-{\rm i}tH}\Omega_v^--U_v(t))\Phi_v\|=O(|v|^{-(1+\epsilon)/2})
\end{equation}
holds as $|v|\rightarrow\infty$.
\end{Lem}

\begin{proof}[Proof of Theorem \ref{the2}]
We define
\begin{equation}
V_v(t,x)=V^{\rm sing}+V^{\rm reg}-V^{\rm reg}(\sinh(2t)v)
\end{equation}
for simplicity. By the same computation as for the proof of \cite[Theorem 2.1]{Is1}, we have
\begin{equation}
|v|({\rm i})[S,p_j]\Phi_v,\Psi_v)=e^{-{\rm i}\int_{-\infty}^\infty V^{\rm reg}(\sinh(2t)v){\rm d}t}\{I(v)+R(v)\}
\end{equation}
where
\begin{align}
I(v)&=|v|\int_{-\infty}^\infty\{(V_v(t,x)U_v(t)(p_j\Phi_0)_v,U_v(t)\Psi_v)\nonumber\\
&\qquad-(V_v(t,x)U_v(t)\Phi_v,U_v(t)(p_j\Psi_0)_v)\}{\rm d}t,\\
R(v)&=|v|\int_{-\infty}^\infty\{((e^{-{\rm i}tH}\Omega_v^--U_v(t))(p_j\Phi_0)_v,V_v(t,x)U_v(t)\Psi_v)\nonumber\\
&\qquad-((e^{-{\rm i}tH}\Omega_v^--U_v(t))\Phi_v,V_v(t,x)U_v(t)(p_j\Psi_0)_v)\}{\rm d}t.
\end{align}
We have
\begin{gather}
\int_{-\infty}^\infty|V^{\rm reg}(\sinh(2t)v)|{\rm d}t\lesssim\int_{-\infty}^\infty\langle \sinh(2t)v\rangle^{-1-\epsilon}{\rm d}t\nonumber\\
=|v|^{-1}\int_{-\infty}^\infty\langle\tau\rangle^{-1-\epsilon}\langle \tau/|v|\rangle^{-1}{\rm d}\tau/2=O(|v|^{-1})
\end{gather}
by assumption \eqref{reg} and changing $\tau=\sinh(2t)|v|$, and have
\begin{equation}
e^{-{\rm i}\int_{-\infty}^\infty V^{\rm reg}(\sinh(2t)v){\rm d}t}\rightarrow1
\end{equation}
as $|v|\rightarrow\infty$. Lemmas \ref{lem5}, \ref{lem7}, and \ref{lem8} imply
\begin{equation}
R(v)\rightarrow0
\end{equation}
as $|v|\rightarrow\infty$. We next consider the term $I(v)$. As in the proof of \cite[Theorem 2.1]{Is1}, we have
\begin{gather}
I(v)=\int_{-\infty}^\infty\langle t/|v|\rangle^{-1}\{(V^{\rm sing}(x+\hat{v}t)e^{-{\rm i}\arcsinh(t/|v|)H_0/2}p_j\Phi_0,e^{-{\rm i}\arcsinh(t/|v|)H_0/2}\Psi_0)\nonumber\\
-(V^{\rm sing}(x+\hat{v}t)e^{-{\rm i}\arcsinh(t/|v|)H_0/2}\Phi_0,e^{-{\rm i}\arcsinh(t/|v|)H_0/2}p_j\Psi_0)\}{\rm d}t/2\nonumber\\
+|v|\int_{-\infty}^\infty({\rm i}(\partial_{x_j}V^{\rm reg})(x+\hat{v}t)e^{-{\rm i}\arcsinh(t/|v|)H_0/2}\Phi_0,e^{-{\rm i}\arcsinh(t/|v|)H_0/2}\Psi_0){\rm d}t/2
\end{gather}
and
\begin{gather}
\int_{|t|<|v|^\sigma}\langle t/|v|\rangle^{-1}(V^{\rm sing}(x+\hat{v}t)e^{-{\rm i}\arcsinh(t/|v|)H_0/2}\Phi_0,e^{-{\rm i}\arcsinh(t/|v|)H_0/2}\Psi_0){\rm d}t\nonumber\\
\rightarrow\int_{-\infty}^\infty(V^{\rm sing}(x+\hat{v}t)\Phi_0,\Psi_0){\rm d}t\nonumber\\
|v|\int_{-\infty}^\infty({\rm i}(\partial_{x_j}V^{\rm reg})(x+\hat{v}t)e^{-{\rm i}\arcsinh(t/|v|)H_0/2}\Phi_0,e^{-{\rm i}\arcsinh(t/|v|)H_0/2}\Psi_0){\rm d}t\nonumber\\
\rightarrow\int_{-\infty}^\infty({\rm i}(\partial_{x_j}V^{\rm reg})(x+\hat{v}t)\Phi_0,\Psi_0){\rm d}t
\end{gather}
as $|v|\rightarrow\infty$ for $\sigma=1/(2+\epsilon)$. To complete the proof, it suffices to prove that
\begin{equation}
\int_{|t|\geqslant|v|^\sigma}\langle t/|v|\rangle^{-1}(V^{\rm sing}(x+\hat{v}t)e^{-{\rm i}\arcsinh(t/|v|)H_0/2}\Phi_0,e^{-{\rm i}\arcsinh(t/|v|)H_0/2}\Psi_0){\rm d}t\rightarrow0\label{the3_1}
\end{equation}
as $|v|\rightarrow\infty$. Using the Mehler formula and $\tanh(\arcsinh t)=t\langle t\rangle^{-1}$, we have
\begin{gather}
\langle t/|v|\rangle^{-1}|(V^{\rm sing}(x+\hat{v}t)e^{-{\rm i}\arcsinh(t/|v|)H_0/2}\Phi_0,e^{-{\rm i}\arcsinh(t/|v|)H_0/2}\Psi_0)|\nonumber\\
\leqslant\langle t/|v|\rangle^{-1}\|V^{\rm sing}((t/|v|)(x+v))\mathscr{F}\mathscr{M}((t/|v|)\langle t/|v|\rangle^{-1}/2)\Phi_0\|\|\Psi_0\|\nonumber\\
\leqslant(R_1+R_2)\|\Psi_0\|
\end{gather}
where
\begin{align}
R_1&=\|V^{\rm sing}((t/|v|)(x+v))\langle(|v|/t)p\rangle^{-2}F(|x|<|v|^{1+\sigma}/(2|t|))\|_\mathscr{B}\nonumber\\
&\quad\times\|\mathscr{F}\mathscr{M}((t/|v|)\langle t/|v|\rangle^{-1}/2)\langle(|v|/t)x\rangle^2\Phi_0\|,\\
R_2&=(|v|/|t|)\|V^{\rm sing}((t/|v|)(x+v))\|\nonumber\\
&\times\|\langle(|v|/t)p\rangle^{-2}F(|x|\geqslant|v|^{1+\sigma}/(2|t|))\mathscr{F}\mathscr{M}((t/|v|)\langle t/|v|\rangle^{-1}/2)\langle(|v|/t)x\rangle^2\Phi_0\|_\infty.
\end{align}
When $|x|<|v|^{1+\sigma}/(2|t|)$ and $|t|\geqslant|v|^\sigma$, clearly $|tx|/|v|<|v|^\sigma/2\leqslant|t|/2$ and $|(t/|v|)(x+v)|>|t|/2$ hold. We therefore have
\begin{gather}
\|V^{\rm sing}((t/|v|)(x+v))\langle(|v|/t)p\rangle^{-2}F(|x|<|v|^{1+\sigma}/(2|t|))\|_\mathscr{B}\nonumber\\
\leqslant\|V^{\rm sing}((t/|v|)(x+v))\langle(|v|/t)p\rangle^{-2}F(|(t/|v|)(x+v)|>|t|/2)\|_\mathscr{B}\nonumber\\
=\|V^{\rm sing}(x)\langle p\rangle^{-2}F(|x|>|t|/2)\|_\mathscr{B}\lesssim t^{-4}\|V^{\rm sing}(x)\langle p\rangle^{-2}(\Delta^2\chi)(2x/\tau)\|_\mathscr{B}
\end{gather}
computing the commutator once more in the same way as \eqref{lem5_1}, and
\begin{equation}
\int_{|t|\geqslant|v|^\sigma}R_1{\rm d}t\lesssim|v|^{2-2\sigma}\|\langle x\rangle^2\Phi_0\|\int_{|v|^\sigma}^\infty t^{-4}{\rm d}t=O(|v|^{-1-2\sigma})
\end{equation}
is obtained noting $|v|/|t|\leqslant|v|^{1-\sigma}$ for $|t|\geqslant|v|^\sigma$. We now finally consider $R_2$. Applying Lemma \ref{lem6} for $m(x)=\langle x\rangle^{-2}$ and $\lambda=|v|/t$, we have
\begin{equation}
\|\langle(|v|/t)p\rangle^{-2}\|_{\mathscr{B}(L^\infty)}\lesssim1.
\end{equation}
Integrating by parts, as in \eqref{lem5_3}, we have
\begin{gather}
\mathscr{F}\mathscr{M}((t/|v|)\langle t/|v|\rangle^{-1}/2)\langle(|v|/t)x\rangle^2\Phi_0\nonumber\\
=\langle x\rangle^{-2}\mathscr{F}\mathscr{M}((t/|v|)\langle t/|v|\rangle^{-1}/2)\langle(|v|/t)x\rangle^2\Phi_0\nonumber\\
+(|v|/t)\langle t/|v|\rangle\langle x\rangle^{-2}x\cdot\mathscr{F}x\mathscr{M}((t/|v|)\langle t/|v|\rangle^{-1}/2)\langle(|v|/t)x\rangle^2\Phi_0\nonumber\\
-{\rm i}\langle x\rangle^{-2}x\cdot\mathscr{F}\mathscr{M}((t/|v|)\langle t/|v|\rangle^{-1}/2)\nabla_x\langle(|v|/t)x\rangle^2\Phi_0.
\end{gather}
If $|x|\geqslant|v|^{1+\sigma}/(2|t|)$ and $|t|\geqslant|v|^\sigma$, clearly $|x|\geqslant|v|/2$ holds. By the Hausdorff--Young inequality and $|v|/|t|\leqslant|v|^{1-\sigma}$ for $|t|\geqslant|v|^\sigma$, we have
\begin{gather}
\|F(|x|\geqslant|v|^{1+\sigma}/(2|t|))\langle x\rangle^{-2}\mathscr{F}\mathscr{M}((t/|v|)\langle t/|v|\rangle^{-1}/2)\langle(|v|/t)x\rangle^2\Phi_0\|_\infty\nonumber\\
\lesssim|v|^{-(2+N)\sigma}\|\langle x\rangle^2\Phi_0\|_1
\end{gather}
performing the integration by parts $N$-times. The second terms is estimated and yields
\begin{gather}
\|F(|x|\geqslant|v|^{1+\sigma}/(2|t|))\langle v/t\rangle\langle x\rangle^{-2}x\cdot\mathscr{F}x\mathscr{M}((t/|v|)\langle t/|v|\rangle^{-1}/2)\langle(|v|/t)x\rangle^2\Phi_0\|_\infty\nonumber\\
\lesssim|v|^{2-(2+N)\sigma}\|\langle x\rangle^2\Phi_0\|_1,
\end{gather}
where we used $(|v|/|t|)\langle t/|v|\rangle=\langle v/t\rangle\lesssim|v|^{1-\sigma}$ for $|t|\geqslant|v|^\sigma$. The third term is
\begin{gather}
\|F(|x|\geqslant|v|^{1+\sigma}/(2|t|))\langle x\rangle^{-2}x\cdot\mathscr{F}\mathscr{M}((t/|v|)\langle t/|v|\rangle^{-1}/2)\nabla_x\langle(|v|/t)x\rangle^2\Phi_0\|_\infty\nonumber\\
\lesssim|v|^{1-(2+N)\sigma}\|\langle x\rangle^2\Phi_0\|_1.
\end{gather}
Combining these estimates and
\begin{equation}
\|V^{\rm sing}((t/|v|)(x+v))\|=(|t|/|v|)^{-n/2}\|V^{\rm sing}\|,
\end{equation}
we have
\begin{equation}
\int_{|t|\geqslant|v|^\sigma}R_2{\rm d}t\lesssim|v|^{3+n/2-(2+N)\sigma}\int_{|v|^\sigma}^\infty t^{-1-n/2}{\rm d}t=O(|v|^{3+n/2-(2+N+n/2)\sigma}).
\end{equation}
We can choose $N\in\mathbb{N}$ such that $3+n/2-(2+N+n/2)\sigma<0$ and hence have \eqref{the3_1} as $|v|\rightarrow\infty$.
\end{proof}
\bigskip
%acknowledgments---------------------------------------------------
\noindent\textbf{Acknowledgments.} This work was supported by JSPS KAKENHI Grant Numbers JP20K03625 and JP21K03279.\\\\
\noindent\textbf{Data Availability.} The data that support the findings of this study are available from the corresponding author upon reasonable request. 

%references---------------------------------------------------

\end{document}